\pdfoutput=1
\documentclass[11pt,a4paper]{article}
\usepackage[T1]{fontenc}
\usepackage[utf8]{inputenc}
\usepackage[UKenglish]{babel}
\usepackage{lmodern}
\usepackage[margin=25mm]{geometry}
\usepackage{amsmath,amssymb,amsthm,booktabs,placeins,graphicx}
\usepackage{microtype}
\usepackage[round,authoryear]{natbib}
\usepackage{authblk}
\usepackage{xurl}
\usepackage[unicode,hidelinks]{hyperref}
\providecommand{\doi}[1]{doi:\space\href{https://doi.org/#1}{\nolinkurl{#1}}}
\newtheorem{theorem}{Theorem}[section]
\newtheorem{proposition}[theorem]{Proposition}
\newtheorem{lemma}[theorem]{Lemma}
\newtheorem{corollary}[theorem]{Corollary}
\theoremstyle{definition}
\newtheorem{definition}[theorem]{Definition}
\newcommand{\E}{\mathbb E}
\newcommand{\TV}{\operatorname{TV}}
\newcommand{\OPT}{\operatorname{OPT}_{\mathrm{on}}}
\newcommand{\err}{\mathcal E}
\newcommand{\ind}{\mathbf 1}
\newcommand{\reg}{\operatorname{Reg}}
\newcommand{\hist}{\mathcal H}

\hypersetup{
 pdftitle={Robust and Learned Online Matching in Growing Trees},
 pdfauthor={Marek Gałązka and Hanna Wdowicka},
 pdfsubject={Online matching under misspecified and learned graph-growth laws},
 pdfkeywords={online matching, growing trees, model misspecification, distributional advice, regret, preferential attachment}
}
\title{\bfseries Robust and Learned Online Matching\\in Growing Trees}
\author[1]{Marek Ga\l{}\k{a}zka}
\author[2]{Hanna Wdowicka}
\affil[1]{Faculty of Mathematics and Computer Science, Adam Mickiewicz University, Pozna\'n, Poland}
\affil[2]{Department of Statistics, Pozna\'n University of Economics and Business, Pozna\'n, Poland}

\date{September 2026}
\begin{document}
\maketitle
\begin{abstract}
We study irrevocable maximum-cardinality matching in trees revealed by successive leaf attachments, with a known horizon and an exogenous growth law that is misspecified or unknown. For deterministic affine attachment forecasts with nonnegative degree reinforcement, the optimal threshold policy loses at most twice the cumulative expected conditional total-variation error relative to an online oracle knowing the actual growth law. This follows from a unit-span property of the Bellman continuation score and has no additional horizon factor. A four-vertex example attains the coefficient two for the specified deterministic policy, and a two-model argument gives a lower bound linear in the model-error budget for arbitrary policies under general misspecification. For uniform--preferential attachment, the local error has an exact expression through the leaf count. When its constant mixture parameter is unknown, we estimate it from the same growing tree and update the threshold policy at geometric times. A parameter-sensitivity bound for individual Bellman prices and uniform degree-moment estimates yield expected regret $O(\sqrt n\,\log^2 n)$, using $O(n^2\log n)$ arithmetic operations and $O(n)$ stored entries. The exact minimax rate remains open.
\end{abstract}

\section{Introduction}

An online matching algorithm on a growing tree must decide whether to accept each new edge before observing subsequent attachments. Accepting an edge occupies both endpoints permanently. Rejecting it preserves their availability but forfeits that edge. The growth mechanism therefore affects the value of waiting: vertices that attract many future children can be worth preserving.

A distribution-specific optimal policy is useful only if its behaviour is stable when the assumed model is inaccurate. In a growing graph this question differs from changing a distribution of independent requests. An attachment changes the state that determines later attachment probabilities, so a local modelling error can also change subsequent graph evolution. A direct finite-horizon perturbation argument can charge an error by the full remaining reward. We ask whether the matching problem admits a sharper structural estimate.

We answer this question for affine attachment forecasts. The true law need not be affine, stationary, or correctly specified. Let $P_t(\cdot\mid\hist_t)$ be its next-parent distribution, where $\hist_t$ is the observed graph history, and let $Q_t(\cdot\mid T_t)$ be the forecast evaluated on that same history. Write
\begin{equation}\label{eq:intro-error}
 \err_n(P,Q)=\sum_{t=2}^{n-1}\E_P\!\left[
 \TV\bigl(P_t(\cdot\mid\hist_t),Q_t(\cdot\mid T_t)\bigr)\right].
\end{equation}
The graph-generation law is exogenous: accepting or rejecting edges does not change it. This makes the error in \eqref{eq:intro-error} independent of the matching policy.

\subsection{Main results}
Our main results are the following.
\begin{enumerate}
\item \emph{Stability under misspecification.} If $\pi_Q$ is the optimal threshold policy for an affine forecast, then
\[
 \OPT(P)-\E_P|M_n^{\pi_Q}|\leq 2\err_n(P,Q).
\]
The proof identifies a conditional continuation score with span at most one and telescopes its Bellman potential. In particular, local error at most $\varepsilon$ costs at most $2(n-2)\varepsilon$, rather than an additional factor proportional to the horizon.
\item \emph{Sharpness and parameter calibration.} A four-vertex example attains the coefficient two for $\pi_Q$ with our tie convention. For arbitrary policies, a two-model construction gives a lower bound $\err_4$. In the uniform--preferential family we express the local error exactly through the number of leaves and obtain a bound of order $n|\theta-\widehat\theta|$. The lower constructions concern general misspecification; they do not establish the optimal parameter dependence inside this family.
\item \emph{Learning the model from one growing tree.} For uniform--preferential attachment with an unknown constant parameter, an inverse-mean leaf-count estimator has root-mean-square error $O(k^{-1/2})$ after $k$ vertices. Updating this estimate at geometric times gives regret $O(\sqrt n\,\log^2 n)$, with $O(n^2\log n)$ arithmetic work and $O(n)$ stored entries. The proof combines sensitivity of the Bellman prices with degree-moment estimates and a fixed true-model potential. We also give the simpler pilot-and-commit bound $O(n^{2/3})$. The exact optimal rate remains open.
\end{enumerate}

\subsection{Structure of the argument}
Section~\ref{sec:model} gives a self-contained derivation of the Bellman representation and threshold policy for known affine ordinary-degree growth kernels. This provides the starting point for the analysis of misspecification and learning an unknown growth parameter. We also explain how the stability argument applies under the rooted plane-oriented convention.

\subsection{Related work}
\citet{acan2019} analyse greedy matching in uniform and preferential attachment graph processes. Their asymptotic greedy benchmarks concern a fixed policy, rather than the cost of acting under an incorrect growth model. Root and self-loop conventions must be distinguished when comparing finite-size values.

A closely related prediction-based result is due to \citet{aamand2022}. They prove stochastic optimality of a predicted-degree priority rule in a bipartite Chung--Lu--Vu model. Their Appendix~D already bounds the loss caused by erroneous priority orders. Our advice instead specifies conditional parent distributions in a graph whose vertex set grows; our error is evaluated along the resulting dependent history. \citet{canonne2025} use distributional advice for minimum-cost metric matching, with Wasserstein error. Their objective, fixed server set, and offline benchmark differ from ours. \citet{choo2024} and \citet{burathep2025} study imperfect advice for bipartite matching under random arrival orders. Their competitive guarantees do not directly address the present stochastic-control benchmark.

Robustness to changing arrival distributions also appears in budgeted online allocation. \citet{zhou2019} use distributionally robust optimisation and periodically update dual prices for drifting user arrivals. Their fixed bidder set, budget constraints, and offline allocation benchmark differ from the irrevocable growing-tree problem considered here.

For adversarial edge arrivals, \citet{buchbinder2017} obtain a $5/9$-competitive algorithm on forests. \citet{jiang2024} study growing trees with free disposal, which permits deletion of accepted edges. We retain irrevocability throughout. Model-error bounds based on value functions are standard; \citet{lobel2024} sharpen general simulation-lemma estimates. Our contribution in this direction is the unit-span property specific to these matching continuation values, rather than a new general perturbation method.

Estimating preferential attachment is also established. \citet{gao2022} study parametric inference, including estimators from a final snapshot. \citet{zhang2022} use counts of degree-one vertices for estimation in a time-varying attachment model. We use the same broad statistical idea and supply an elementary finite-sample bound for our mixture, sufficient to obtain an online matching guarantee.

Learning with exogenous dynamics is studied more generally by \citet{maran2026} in finite-state episodic Markov decision processes. Their state-space framework and our single-trajectory growing-tree model have different complexity parameters. The geometric-update analysis here uses the sensitivity of individual matching prices and moments of the evolving degree sequence.

\section{Model and Bellman structure}\label{sec:model}

Fix an integer horizon $n\geq2$. The process starts from vertices $1,2$ and the edge $\{1,2\}$. At time $2$ this edge is offered once. For $t=2,\ldots,n-1$, vertex $t+1$ arrives with a single edge to a parent $v\in[t]$. Let $T_t$ be the labelled tree before this arrival and $d_t(v)$ its ordinary degree. Thus $\sum_{v\leq t}d_t(v)=2(t-1)$. The algorithm observes the new edge and immediately accepts or rejects it. The observed graph history includes every revealed edge, whether accepted or rejected. Acceptance is feasible precisely when the parent is unmatched. We call the policy that accepts every feasible edge Greedy. Decisions are irrevocable and do not affect any parent distribution.

A policy is nonanticipating and may use internal randomness independent of the graph-generation randomness. The horizon $n$ is known. Under a true law $P$, define
\[
 \OPT(P)=\sup_{\pi\text{ nonanticipating}}\E_P|M_n^\pi|.
\]
The policy in this supremum knows $P$ but not its future realisation. All regrets in the paper use this benchmark. They are not competitive ratios against an offline optimum.

\begin{definition}[Affine forecast]\label{def:affine}
A forecast $Q$ is specified before the process begins by deterministic coefficients $a_t,\beta_t$, $2\leq t<n$, such that
\[
 \begin{gathered}
 q_t(d)=a_t+\beta_t d,\qquad \beta_t\geq0,\qquad ta_t+2(t-1)\beta_t=1,\\
 0\leq q_t(d)\leq1\qquad(1\leq d<t).
 \end{gathered}
\]
It assigns $Q_t(v\mid T_t)=q_t(d_t(v))$.
\end{definition}

Examples include uniform attachment, attachment proportional to $d+\delta_t$ with deterministic $\delta_t>-1$, and
\begin{equation}\label{eq:mixture}
 q_t^{\widehat\theta_t}(d)=\frac{1-\widehat\theta_t}{t}
       +\frac{\widehat\theta_t d}{2(t-1)},\qquad \widehat\theta_t\in[0,1].
\end{equation}
The schedule in \eqref{eq:mixture} may vary with time but is fixed in advance. The true law $P$ may be any distribution on these graph-growth histories, with arbitrary history-dependent conditional parent probabilities. There is no assumption that $P$ belongs to the forecast family.

Let $U$ be the unmatched vertices after the decision at time $t$, and let $V_t^Q(T,U)$ denote the maximum expected number of additional accepted edges under $Q$. The following proposition establishes the Bellman representation for Definition~\ref{def:affine}; its proof is included in full.

\begin{proposition}[Separable value and threshold policy]\label{prop:bellman}
For every legal state,
\begin{equation}\label{eq:value}
 V_t^Q(T,U)=c_t+\sum_{v\in U}b_t(d_t(v)),
\end{equation}
where $c_n=0$, $b_n(d)=0$, and, backwards for $2\leq t<n$,
\begin{align}
 c_t&=c_{t+1}+b_{t+1}(1),\label{eq:c}\\
 b_t(d)&=(1-q_t(d))b_{t+1}(d)
  +q_t(d)\max\{b_{t+1}(d+1),1-b_{t+1}(1)\}.
 \label{eq:b}
\end{align}
Moreover, $0\leq b_t(d)\leq1$ and $b_t(d)$ is nondecreasing in $d$. An optimal policy accepts a feasible arrival at time $t+1$ exactly when
\begin{equation}\label{eq:accept}
 b_{t+1}(d+1)+b_{t+1}(1)\leq1,
\end{equation}
where $d$ is the parent's degree before insertion. It accepts the seed edge if $2b_2(1)\leq1$. We use these weak inequalities as the tie convention defining $\pi_Q$.
\end{proposition}

\begin{proof}
Assume \eqref{eq:value} at time $t+1$ and put
\[
 C=c_{t+1}+b_{t+1}(1)+\sum_{u\in U}b_{t+1}(d_t(u)).
\]
If the parent is matched, the continuation value is $C$. If it is a free vertex of degree $d$, rejection gives $C+b_{t+1}(d+1)-b_{t+1}(d)$, whereas acceptance gives $C+1-b_{t+1}(1)-b_{t+1}(d)$. Maximisation and averaging over the parent give \eqref{eq:c}--\eqref{eq:b}; their comparison gives \eqref{eq:accept}. The terminal condition starts the induction.

The range $[0,1]$ is preserved because \eqref{eq:b} is a convex combination of numbers in this range. For monotonicity, set $A(d)=b_{t+1}(d)$ and $B(d)=\max\{b_{t+1}(d+1),1-b_{t+1}(1)\}$. Both are nondecreasing and $B(d)\geq A(d)$. Since $q_t(d)$ is nondecreasing, $(1-q_t(d))A(d)+q_t(d)B(d)$ is nondecreasing as well. Finally, accepting the seed gives $1+c_2$, and rejecting it gives $c_2+2b_2(1)$.
\end{proof}

Consequently,
\begin{equation}\label{eq:initial}
 V_{\mathrm{init}}^Q=c_2+\max\{1,2b_2(1)\}.
\end{equation}
The acceptance set in \eqref{eq:accept} is a prefix of the degrees. Its threshold can be stored for every $t$. Computing the coefficients requires $O(n^2)$ arithmetic operations; two coefficient rows and the stored thresholds require $O(n)$ entries. Maintaining degrees and matched flags then takes $O(1)$ work per arrival. These are arithmetic-operation bounds, not bounds on exact rational bit complexity. The formula holds for every legal state, including a state reached under a different growth law.

\section{Stability under a misspecified growth law}\label{sec:stability}

We use $\TV(p,q)=\frac12\sum_v|p(v)-q(v)|$ and the error in \eqref{eq:intro-error}. In particular, $Q$ is evaluated on the actual history drawn under $P$.

\begin{lemma}[Unit-span continuation]\label{lem:span}
Fix a state $(T_t,U)$ and let $G_t(v)$ be the maximum of immediate reward plus $V_{t+1}^Q$ after parent $v$ is revealed. Then
\begin{equation}\label{eq:G}
 \begin{aligned}
 G_t(v)&=C+\ind_{\{v\in U\}}h_t(d_t(v)),\\
 h_t(d)&=\max\{b_{t+1}(d+1),1-b_{t+1}(1)\}-b_{t+1}(d).
 \end{aligned}
\end{equation}
with $C$ independent of $v$ and $0\leq h_t(d)\leq1$. In particular,
$\max_vG_t(v)-\min_vG_t(v)\leq1$.
\end{lemma}

\begin{proof}
The expression is the action comparison in Proposition~\ref{prop:bellman}. Monotonicity makes the first difference nonnegative, and the maximum and the subtracted coefficient both belong to $[0,1]$.
\end{proof}

\begin{theorem}[Model-error bound]\label{thm:robust}
For every affine forecast $Q$ and every exogenous true law $P$,
\begin{equation}\label{eq:robust}
 0\leq\reg_n(P,Q):=\OPT(P)-\E_P|M_n^{\pi_Q}|
 \leq2\err_n(P,Q).
\end{equation}
More precisely,
\begin{equation}\label{eq:two-sided}
 \OPT(P)\leq V_{\mathrm{init}}^Q+\err_n(P,Q),\qquad
 \E_P|M_n^{\pi_Q}|\geq V_{\mathrm{init}}^Q-\err_n(P,Q).
\end{equation}
\end{theorem}

\begin{proof}
For a function of span at most one, the difference between its expectations under $p$ and $q$ is at most $\TV(p,q)$ in absolute value. Apply this to Lemma~\ref{lem:span}, conditionally on the current graph and matching state. Under $Q$, the Bellman equation says $\E_{Q_t}G_t=V_t^Q$.

Let $R_t$ count edges accepted up to time $t$. Any policy chooses an action with reward plus next potential at most $G_t(v)$. Thus, under $P$,
\begin{equation}\label{eq:increment}
 \E_P[R_{t+1}+V_{t+1}^Q-R_t-V_t^Q\mid\hist_t,U_t]
 \leq e_t(\hist_t),
\end{equation}
where $e_t(\hist_t)=\TV(P_t(\cdot\mid\hist_t),Q_t(\cdot\mid T_t))$. Internal policy randomness gives no additional information about future graph growth. For $\pi_Q$ the selected action attains $G_t(v)$, so the same conditional expectation is also at least $-e_t(\hist_t)$.

For every initial seed decision, $R_2+V_2^Q\leq V_{\mathrm{init}}^Q$, with equality for $\pi_Q$. Sum \eqref{eq:increment}, use $V_n^Q=0$, and take expectations. The distribution of graph histories is the same under every policy, so the error sum is always \eqref{eq:intro-error}. This proves \eqref{eq:two-sided}; subtracting gives \eqref{eq:robust}. The lower bound zero follows from the definition of $\OPT$.
\end{proof}

\begin{corollary}[Conditional suffix bound]\label{cor:suffix}
Fix a history at time $k$ and an unmatched set $U$. Starting from this state, the loss of $\pi_Q$ relative to the optimal continuation under $P$ is at most
\[
 2\E_P\!\left[\sum_{t=k}^{n-1}e_t(\hist_t)\,\middle|\,\hist_k\right].
\]
The forecast may be selected using $\hist_k$, provided its future schedule is then fixed.
\end{corollary}

\begin{proof}
Condition on $\hist_k$ and telescope from the common potential $V_k^Q(T_k,U)$ instead of \eqref{eq:initial}.
\end{proof}

This conditional formulation is essential for learning. It does not justify replacing the forecast arbitrarily after each new observation: such replacements would change the potential being telescoped. The estimate also relies on exogenous graph growth. If matching decisions influence future parent probabilities, a policy-independent error sum need not exist.

\paragraph*{The rooted plane-oriented convention.}
The stability argument also applies to the rooted plane-oriented convention: start from a single root and use attachment weight $w_t(v)=1+d_t^+(v)$, where $d_t^+$ is the number of children. Here $\sum_vw_t(v)=2t-1$. For deterministic affine weight forecasts $q_t(w)=a_t+\beta_t w$, impose $ta_t+(2t-1)\beta_t=1$ with $\beta_t\geq0$ and $0\leq q_t(w)\leq1$ for every integer $1\leq w\leq t$. Replace degree by weight in \eqref{eq:b}, run the recurrence down to $t=1$, and use initial value $c_1+b_1(1)$. A parent increases its weight by one and a newborn has weight one, so \eqref{eq:G} and its unit-span proof are unchanged. Thus \eqref{eq:robust} holds with the error sum starting at $t=1$. The explicit leaf formulas and learning constants below refer to the seed-edge convention.

\section{Sharpness for general misspecification}\label{sec:sharpness}

\begin{proposition}[A tight plug-in example]\label{prop:sharp}
For every $0<\varepsilon\leq1/3$, there is a true law $P_+$ on four-vertex trees and a uniform forecast $Q$ such that
\[
 \err_4(P_+,Q)=\varepsilon,\qquad \reg_4(P_+,Q)=2\varepsilon.
\]
\end{proposition}

\begin{proof}
Vertex $3$ chooses its parent $i\in\{1,2\}$ uniformly. Write $j$ for the other seed endpoint. At the last arrival set
\begin{equation}\label{eq:lowerlaw}
 P_+(i\mid\hist_3)=\tfrac13,\quad
 P_+(j\mid\hist_3)=\tfrac13+\varepsilon,\quad
 P_+(3\mid\hist_3)=\tfrac13-\varepsilon.
\end{equation}
The only forecast error occurs at this last step and equals $\varepsilon$.
Under the uniform forecast, $b_3(1)=b_3(2)=1/3$ and $b_2(1)=1/2$. Our tie convention therefore accepts the seed. Its only later opportunity is the edge from $4$ to $3$, giving value $4/3-\varepsilon$.

Rejecting the seed, accepting $\{i,3\}$, and accepting $\{j,4\}$ if offered gives $4/3+\varepsilon$. This is optimal: conditional on rejecting the seed, rejecting the next edge can give at most one edge in total. Conditional on accepting the seed, the value is the one already computed. The difference is $2\varepsilon$.
\end{proof}

\begin{proposition}[Linear dependence cannot be removed]\label{prop:minimax}
For every $0<\varepsilon\leq1/3$, let $Q$ be the uniform forecast and let $P_+,P_-$ be the pair defined below. For every online policy given $Q$ but not the identity of the true law, some $P\in\{P_+,P_-\}$ satisfies
\[
 \err_4(P,Q)=\varepsilon,\qquad
 \OPT(P)-\E_P|M_4|\geq\varepsilon.
\]
This holds for randomised policies.
\end{proposition}

\begin{proof}
Obtain $P_-$ by reversing the signs of the two perturbations in \eqref{eq:lowerlaw}. Both laws have the same history distribution before the initial decision. If the algorithm accepts the seed with probability $r$, its regret under $P_+$ is at least $2r\varepsilon$, and under $P_-$ at least $2(1-r)\varepsilon$, even allowing the best continuation after that decision. Their maximum is at least $\varepsilon$.
\end{proof}

These statements distinguish two notions of sharpness. Proposition~\ref{prop:sharp} establishes the coefficient for the specified deterministic plug-in policy; Proposition~\ref{prop:minimax} leaves a factor-two gap for the best robust policy. Both true laws distinguish two degree-one vertices through their histories. Neither is a constant-parameter instance of \eqref{eq:mixture}.

\section{Calibration in the uniform--preferential family}\label{sec:calibration}

Let $P_\theta$ denote \eqref{eq:mixture} with a fixed $\theta\in[0,1]$, and let $L_t$ be the number of degree-one vertices in $T_t$. For a forecast with constant parameter $\widehat\theta$, write $\reg_n(\theta,\widehat\theta)=\reg_n(P_\theta,P_{\widehat\theta})$.

\begin{proposition}[Leaf-count error identity]\label{prop:tv}
For every tree $T_t$,
\begin{equation}\label{eq:tv-leaves}
 \TV(P_\theta(\cdot\mid T_t),P_{\widehat\theta}(\cdot\mid T_t))
 =|\theta-\widehat\theta|\frac{L_t(t-2)}{2t(t-1)}.
\end{equation}
Consequently, with $\ell_t(\theta)=\E_\theta L_t$ and $H_m=\sum_{j=1}^m j^{-1}$,
\begin{align}
 \reg_n(\theta,\widehat\theta)
 &\leq |\theta-\widehat\theta|\sum_{t=2}^{n-1}
       \ell_t(\theta)\frac{t-2}{t(t-1)}\label{eq:leaf-regret}\\
 &\leq (n-2H_{n-1})|\theta-\widehat\theta|
 \leq(n-2)|\theta-\widehat\theta|.\label{eq:param-regret}
\end{align}
\end{proposition}

\begin{proof}
The difference of the kernels is $(\theta-\widehat\theta)(d/(2(t-1))-1/t)$. For $t>2$, the bracket is negative exactly at degree one and positive at degrees at least two. Summing the negative part gives \eqref{eq:tv-leaves}; at $t=2$ both sides vanish. Theorem~\ref{thm:robust} yields \eqref{eq:leaf-regret}. Since $L_t\leq t-1$ for $t\geq3$, the remaining sum is at most $\sum_{t=2}^{n-1}(t-2)/t=n-2H_{n-1}$.
\end{proof}

The exact expected-error bound is computable in linear time because
\begin{equation}\label{eq:leaf-mean}
 \ell_2(\theta)=2,\qquad
 \ell_{t+1}(\theta)=\bigl(1-\alpha_t(\theta)\bigr)\ell_t(\theta)+1,
 \quad \alpha_t(\theta)=\frac{1-\theta}{t}+\frac{\theta}{2(t-1)}.
\end{equation}
Indeed, the new vertex adds a leaf, and a leaf parent ceases to be a leaf. The same argument allows distinct deterministic true and forecast parameter schedules; the summand in \eqref{eq:leaf-regret} then contains $|\theta_t-\widehat\theta_t|$.

We will use the uniform conditional consequence of Corollary~\ref{cor:suffix}: for any history and unmatched set at time $k$, the continuation loss under parameter $\widehat\theta$ is at most
\begin{equation}\label{eq:suffix-param}
 (n-k)|\theta-\widehat\theta|.
\end{equation}
This follows pathwise from twice the right-hand side of \eqref{eq:tv-leaves} being at most $|\theta-\widehat\theta|$ at every step.

\section{Learning without external advice}\label{sec:learning}

Assume the true law is $P_\theta$ for an unknown fixed $\theta\in[0,1]$. Graph observations are available regardless of the matching decisions. We first obtain a uniform finite-sample estimator from a single leaf count and then apply \eqref{eq:suffix-param}.

\begin{lemma}[Leaf-count concentration and sensitivity]\label{lem:leaves}
For $k\geq4$, every $\theta\in[0,1]$, and every $u\geq0$,
\begin{align}
 \operatorname{Var}_\theta(L_k)&\leq(k-2)/4,\label{eq:var}\\
 \Pr_\theta\bigl(|L_k-\ell_k(\theta)|\geq u\bigr)
 &\leq2\exp\bigl(-2u^2/(k-2)\bigr),\label{eq:tail}\\
 \ell'_k(\theta)&\geq
 \frac{k-2}{8}-\frac{H_{k-2}}{4(k-1)}\geq\frac{k-3}{8}.\label{eq:derivative}
\end{align}
In particular, $\ell_k$ is strictly increasing on $[0,1]$.
\end{lemma}

\begin{proof}
Let $Y_{t+1}$ indicate that the next parent is a leaf. Its conditional mean is $\alpha_t(\theta)L_t$, and $L_{t+1}=L_t+1-Y_{t+1}$. Subtracting \eqref{eq:leaf-mean}, we obtain
\[
 X_{t+1}=(1-\alpha_t(\theta))X_t+\xi_{t+1},\qquad
 X_t=L_t-\ell_t(\theta),\quad
 \xi_{t+1}=\alpha_t(\theta)L_t-Y_{t+1}.
\]
The $\xi_{t+1}$ are martingale differences of conditional variance at most $1/4$ and conditional range length one. Unrolling the recursion expresses $X_k$ as a sum of $k-2$ such differences with deterministic weights
\[
 w_{s,k}=\prod_{j=s+1}^{k-1}(1-\alpha_j(\theta))\in[0,1],\qquad 2\leq s<k.
\]
Orthogonality gives \eqref{eq:var}. The conditional exponential-moment bound for a centred variable of range length $w_{s,k}$ is $\exp(\lambda^2w_{s,k}^2/8)$. Iteration and Chernoff's bound give \eqref{eq:tail}, since $\sum_sw_{s,k}^2\leq k-2$.

For sensitivity, $\alpha_t(\theta)\leq1/t$ implies $\ell_t(\theta)\geq t/2$ by induction from $\ell_2=2$. Differentiation of \eqref{eq:leaf-mean}, with $J_t=\ell'_t(\theta)$, gives
\[
 J_{t+1}=(1-\alpha_t(\theta))J_t
       +\frac{t-2}{2t(t-1)}\ell_t(\theta),\qquad J_2=0.
\]
Thus $J_t\geq0$ and
\[
 tJ_{t+1}\geq(t-1)J_t+\frac{t(t-2)}{4(t-1)}
 =(t-1)J_t+\frac14\left(t-1-\frac1{t-1}\right).
\]
Summing for $t=2,\ldots,k-1$ proves the first inequality in \eqref{eq:derivative}. The second follows from $2H_{k-2}\leq k-1$ for $k\geq4$.
\end{proof}

Define the ideal clipped inverse estimator
\begin{equation}\label{eq:estimator}
 \widetilde\theta_k=
 \begin{cases}
 0,&L_k\leq\ell_k(0),\\
 \ell_k^{-1}(L_k),&\ell_k(0)<L_k<\ell_k(1),\\
 1,&L_k\geq\ell_k(1).
 \end{cases}
\end{equation}
For computation, use bisection to obtain $\widehat\theta_k\in[0,1]$ with
$|\widehat\theta_k-\widetilde\theta_k|\leq\eta$. Each evaluation of \eqref{eq:leaf-mean} takes $O(k)$ operations, so bisection takes $O(k\log(1/\eta))$ operations for $0<\eta<1$.

\begin{corollary}[Root-mean-square estimation error]\label{cor:estimation}
Uniformly over $\theta\in[0,1]$,
\begin{equation}\label{eq:estimation-bound}
 \left(\E_\theta|\widehat\theta_k-\theta|^2\right)^{1/2}
 \leq\min\left\{1,\frac{4\sqrt{k-2}}{k-3}+\eta\right\}.
\end{equation}
\end{corollary}

\begin{proof}
Projection of $L_k$ onto $[\ell_k(0),\ell_k(1)]$ cannot increase its distance from $\ell_k(\theta)$. The inverse has Lipschitz constant at most $8/(k-3)$ by \eqref{eq:derivative}. Therefore
\[
 |\widetilde\theta_k-\theta|\leq\frac{8}{k-3}|L_k-\ell_k(\theta)|.
\]
Take $L^2$ norms and use \eqref{eq:var}. The triangle inequality in $L^2$ adds at most $\eta$ for bisection, and $|\widehat\theta_k-\theta|\leq1$ gives the truncation. The same expression bounds the expected absolute error by Cauchy--Schwarz.
\end{proof}

\paragraph*{The pilot-and-commit policy.}
Fix $4\leq k\leq n$. Reject every edge through time $k$, including the seed edge. Compute \eqref{eq:estimator} to precision $\eta$ from $L_k$. Compute the threshold policy for $P_{\widehat\theta_k}$ with terminal horizon $n$, and use it on subsequent arrivals. All existing vertices are free at the switch. The estimation phase is passive: its observations do not depend on the rejected decisions.

\begin{theorem}[An unknown-parameter guarantee]\label{thm:learning}
For $n\geq k\geq4$ and $0<\eta<1$, the policy $A_{k,\eta}$ satisfies, uniformly over $\theta\in[0,1]$,
\begin{equation}\label{eq:learning-regret}
 \OPT(P_\theta)-\E_\theta|M_n^{A_{k,\eta}}|
 \leq\left\lfloor\frac{k}{2}\right\rfloor
 +(n-k)\min\left\{1,\frac{4\sqrt{k-2}}{k-3}+\eta\right\}.
\end{equation}
Choosing $k=\min\{n,\max\{4,\lceil n^{2/3}\rceil\}\}$ and $\eta=1/n$ yields $O(n^{2/3})$ expected regret. The policy uses $O(n^2)$ arithmetic operations in total and $O(n)$ stored entries.
\end{theorem}

\begin{proof}
Let $V_k^\theta(T,U)$ be the true-model optimal continuation. An oracle policy has accepted at most $\lfloor k/2\rfloor$ edges by time $k$. Additional free vertices cannot reduce its continuation value; this also follows directly from the nonnegative coefficients in Proposition~\ref{prop:bellman}. Hence
\begin{equation}\label{eq:pilot-oracle}
 \OPT(P_\theta)\leq\lfloor k/2\rfloor+
 \E_\theta V_k^\theta(T_k,[k]).
\end{equation}
Condition on the observed pilot history. The estimate is now fixed, all vertices are free, and \eqref{eq:suffix-param} bounds the continuation loss by $(n-k)|\theta-\widehat\theta_k|$. Average over pilot histories, use \eqref{eq:pilot-oracle}, and apply Corollary~\ref{cor:estimation}. The stated choice balances $k$ and $n/\sqrt{k}$. The estimator requires $O(k\log n)$ work and the threshold preprocessing $O(n^2)$; after that preprocessing each arrival takes $O(1)$ work.
\end{proof}

The theorem is an existence guarantee with an explicit polynomial-time policy. Discarding all pilot edges is convenient for comparison with the oracle, but is not claimed to be statistically or algorithmically optimal. Nor does the theorem establish a lower bound of order $n^{2/3}$.

\section{Learning with geometric updates}\label{sec:epoch}

The $n^{2/3}$ rate in Theorem~\ref{thm:learning} comes from a particular
pilot comparison. Observing the tree does not require rejecting its edges.
We now allow the estimate to change at geometric times and obtain a
uniformly smaller asymptotic regret bound. The proof controls the
sensitivity of individual Bellman prices and telescopes the fixed
true-model value function. It therefore does not require telescoping
different forecast value functions.

\subsection{Sensitivity of the Bellman prices}

Write $b_s^\theta$ for the Bellman coefficients of
Proposition~\ref{prop:bellman} under the kernel $P_\theta$.
For $2\leq s\leq n$, define
\begin{equation}\label{eq:sensitivity-factor}
 K_{s,n}=\prod_{r=s}^{n-1}\left(1+\frac{1}{2(r-1)}\right)-1,
 \qquad K_{n,n}=0.
\end{equation}
An empty product equals one. For $s\geq3$,
\begin{equation}\label{eq:factor-estimate}
 K_{s,n}\leq\sqrt{\frac{n-2}{s-2}}-1.
\end{equation}
Indeed, bound the logarithm of the product by
$\frac12\sum_{r=s}^{n-1}(r-1)^{-1}
\leq\frac12\log((n-2)/(s-2))$.

\begin{lemma}[Parameter sensitivity]\label{lem:sensitivity}
For every $\theta,\lambda\in[0,1]$, $2\leq s\leq n$, and legal degree $d$,
\begin{equation}\label{eq:price-sensitivity}
 |b_s^\theta(d)-b_s^\lambda(d)|
 \leq|\theta-\lambda|(d+1)K_{s,n}.
\end{equation}
Consequently, if an available edge with parent degree $d$ is offered
after time $t$, the loss in immediate reward plus the true optimal
continuation caused by the $\lambda$-threshold decision is at most
\begin{equation}\label{eq:local-action-loss}
 |\theta-\lambda|(d+4)K_{t+1,n}.
\end{equation}
\end{lemma}

\begin{proof}
Write $\delta=|\theta-\lambda|$. The mixture probabilities satisfy
\[
 q_s^\theta(d)\leq\frac{d+1}{2(s-1)},\qquad
 |q_s^\theta(d)-q_s^\lambda(d)|
 \leq\frac{\delta d}{2(s-1)}.
\]
For the second inequality, the derivative of the kernel is
$d/(2(s-1))-1/s$; its absolute value obeys the stated bound also when
$d=1$.

Set $F^\gamma(d)=\max\{b_{s+1}^\gamma(d+1),
1-b_{s+1}^\gamma(1)\}$ for $\gamma\in\{\theta,\lambda\}$, and write
$e_s(d)=|b_s^\theta(d)-b_s^\lambda(d)|$.
Subtracting the two Bellman recurrences, using
$0\leq F^\lambda(d)-b_{s+1}^\lambda(d)\leq1$ from
Lemma~\ref{lem:span}, and taking absolute values gives
\begin{align*}
 e_s(d)&\leq(1-q_s^\theta(d))e_{s+1}(d)
                  +\frac{\delta d}{2(s-1)}\\
 &\quad+q_s^\theta(d)\max\{e_{s+1}(d+1),e_{s+1}(1)\}.
\end{align*}
Assume \eqref{eq:price-sensitivity} at $s+1$ and put $K=K_{s+1,n}$.
The right-hand side is at most
\begin{align*}
 &\delta\left((d+1)K+q_s^\theta(d)K+\frac{d}{2(s-1)}\right)\\
 &\qquad\leq\delta(d+1)\left[
 \left(1+\frac{1}{2(s-1)}\right)K+\frac{1}{2(s-1)}\right].
\end{align*}
The bracket equals $K_{s,n}$. The zero terminal prices start the
backward induction.

The true advantage of accepting over rejecting an available edge is
\[
 \Delta_t(d;\theta)=1-b_{t+1}^\theta(d+1)-b_{t+1}^\theta(1).
\]
If the $\lambda$-decision is suboptimal under $\theta$, the loss is
$|\Delta_t(d;\theta)|$ and is at most
$|\Delta_t(d;\theta)-\Delta_t(d;\lambda)|$. The same bound is zero or
nonnegative when the decisions agree. Applying
\eqref{eq:price-sensitivity} at degrees $d+1$ and $1$ proves
\eqref{eq:local-action-loss}, including the prescribed tie convention.
\end{proof}

\subsection{Degree moments and dependent estimation error}

Let $S_t=\sum_{v\leq t}d_t(v)^2$, and let $D_t$ be the degree of the
parent selected by vertex $t+1$, before its degree increases. Write
$\mu_t=\E_\theta[D_t\mid\hist_t]$.

\begin{lemma}[Uniform degree-moment bounds]\label{lem:degree-moments}
For every $t\geq2$ and $\theta\in[0,1]$,
\begin{align}
 \E_\theta S_t&\leq2(t-1)H_{t-1},\label{eq:degree-first}\\
 \E_\theta S_t^2&\leq t(t-1)(6H_{t-1}^2-4),\label{eq:degree-second}\\
 \bigl(\E_\theta\mu_t^2\bigr)^{1/2}
 &\leq\sqrt{3}\,H_{t-1}.\label{eq:parent-rms}
\end{align}
\end{lemma}

\begin{proof}
For a fixed tree, size biasing the degrees increases their first and
second moments. More explicitly, the nonnegative covariance of $d$
and $d^j$ under the uniform vertex distribution gives
\[
 \frac1t\sum_vd_v^j
 \leq\frac{\sum_vd_v^{j+1}}{2(t-1)},\qquad j=1,2.
\]
Both conditional moments in the mixture are therefore bounded by
those at $\theta=1$. Since $d_v\leq t-1$,
\[
 \mu_t\leq\frac{S_t}{2(t-1)},\qquad
 \E_\theta[D_t^2\mid\hist_t]
 \leq\frac{\sum_vd_v^3}{2(t-1)}\leq\frac{S_t}{2}.
\]
The update is $S_{t+1}=S_t+2D_t+2$. With
$m_t=\E_\theta S_t$ and $u_t=\E_\theta S_t^2$, it follows that
\begin{align*}
 m_{t+1}&\leq\left(1+\frac1{t-1}\right)m_t+2,\\
 u_{t+1}&\leq\left(1+\frac2{t-1}\right)u_t
       +\left(6+\frac4{t-1}\right)m_t+4.
\end{align*}
Dividing the first recurrence by $t$ and using $m_2=2$ gives
\eqref{eq:degree-first}. In the second recurrence,
\[
 \left(6+\frac4{t-1}\right)m_t+4
 \leq(12t-4)H_{t-1}+4\leq12tH_{t-1}.
\]
Consequently,
\[
 \frac{u_{t+1}}{t(t+1)}
 \leq\frac{u_t}{t(t-1)}+\frac{12H_{t-1}}{t+1}
 \leq\frac{u_t}{t(t-1)}+6(H_t^2-H_{t-1}^2).
\]
Telescope from $u_2=4$ to obtain \eqref{eq:degree-second}.
Finally,
\[
 \E_\theta\mu_t^2
 \leq\frac{t}{4(t-1)}(6H_{t-1}^2-4)
 \leq3H_{t-1}^2,
\]
which proves \eqref{eq:parent-rms}.
\end{proof}

The squared-error estimate in Corollary~\ref{cor:estimation} and
Lemma~\ref{lem:degree-moments} allow us to use Cauchy--Schwarz when
the parameter estimate and the current degrees are dependent.
No independence between these two quantities is assumed.

\subsection{The policy and its regret}

For $n\geq4$, the \emph{geometric-update policy} $A^{\mathrm{geo}}$
accepts the seed and uses Greedy through time $4$. Immediately after
processing the edge that creates vertex $k$, for
$k\in\{4,8,16,\ldots\}$ with $k<n$, it computes
$\widehat\theta_k$ from $L_k$ with bisection error at most $1/n$.
It then recomputes the threshold schedule for this estimate with
the original terminal horizon $n$. The schedule is used from
vertex $k+1$ until the next update. Accepted edges are retained
throughout. There is no phase in which all edges are discarded.

For $t\geq4$, let
$\kappa(t)=2^{\lfloor\log_2t\rfloor}$ and define
\begin{equation}\label{eq:epoch-error}
 \epsilon_{k,n}=
 \min\left\{1,\frac{4\sqrt{k-2}}{k-3}+\frac1n\right\},
 \qquad k\geq4.
\end{equation}
Thus $\kappa(t)$ is the most recent update time.

\begin{theorem}[Learning with geometric updates]\label{thm:epoch}
For every $n\geq4$, uniformly over the unknown constant parameter $\theta\in[0,1]$,
\begin{align}
 \OPT(P_\theta)-\E_\theta|M_n^{A^{\mathrm{geo}}}|
 &\leq3+\sum_{t=4}^{n-1}
    \epsilon_{\kappa(t),n}K_{t+1,n}
          (4+\sqrt3\,H_{t-1})\label{eq:epoch-finite}\\
 &=O(\sqrt n\,\log^2 n).\label{eq:epoch-rate}
\end{align}
The algorithm uses $O(n^2\log n)$ arithmetic operations in total and
$O(n)$ stored entries.
\end{theorem}

\begin{proof}
Use the true-model optimal value function $V_t^\theta$ as a fixed
potential throughout the execution. Conditional on the selected
parent, let $g_t$ be the difference between the maximum of immediate
reward plus true continuation and the value of the selected action.
It is zero when the parent is matched. Otherwise
\[
 g_t=\max\{0,\Delta_t(D_t;\theta)\}
           -\ind_{\{\text{accept}\}}\Delta_t(D_t;\theta).
\]
The true Bellman equation and telescoping give the exact identity
\[
 \OPT(P_\theta)-\E_\theta|M_n^{A^{\mathrm{geo}}}|
   =g_{\mathrm{seed}}+\sum_{t=2}^{n-1}\E_\theta g_t.
\]
Here $g_{\mathrm{seed}}=\max\{1,2b_2^\theta(1)\}-1$.
All these gaps are nonnegative and at most one. The seed and the
arrivals at times $3$ and $4$ contribute at most three.

For $t\geq4$, put $k=\kappa(t)$ and
$\delta_k=|\widehat\theta_k-\theta|$.
The estimate is measurable before the next parent is drawn.
Lemma~\ref{lem:sensitivity}, conditional expectation, and
Cauchy--Schwarz yield
\begin{align*}
 \E_\theta g_t
 &\leq K_{t+1,n}\E_\theta[\delta_k(\mu_t+4)]\\
 &\leq K_{t+1,n}
       \bigl(\E_\theta\delta_k^2\bigr)^{1/2}
       \left((\E_\theta\mu_t^2)^{1/2}+4\right)\\
 &\leq K_{t+1,n}\epsilon_{k,n}(4+\sqrt3\,H_{t-1}).
\end{align*}
The last line uses Corollary~\ref{cor:estimation} and
Lemma~\ref{lem:degree-moments}. Exogenous growth ensures that the
estimator has the stated distribution under this policy. Summing
proves \eqref{eq:epoch-finite}.

Since $t/2<\kappa(t)\leq t$ and $\kappa(t)\geq4$,
\[
 \epsilon_{\kappa(t),n}\leq\frac{16}{\sqrt t}+\frac1n,
 \qquad K_{t+1,n}\leq2\sqrt{\frac nt}.
\]
Using $H_{t-1}\leq1+\log t$, the sum in
\eqref{eq:epoch-finite} is
\[
 O\left(\sqrt n\sum_{t=4}^{n-1}\frac{1+\log t}{t}
       +\frac1{\sqrt n}\sum_{t=4}^{n-1}
                                  \frac{1+\log t}{\sqrt t}\right)
 =O(\sqrt n\,\log^2 n).
\]
There are $O(\log n)$ updates. Each threshold computation costs
$O(n^2)$ arithmetic operations and uses two coefficient rows and
one threshold schedule. Estimation costs $O(k\log n)$ at time $k$,
and the sum of update times is $O(n)$. Degrees, matched flags,
the current coefficient rows, and thresholds require $O(n)$
stored entries. Between updates each arrival takes $O(1)$ work.
\end{proof}

Define the minimax oracle regret over policies that know $n$ but
do not know $\theta$ by
\[
 \mathfrak R_n=\inf_A\sup_{\theta\in[0,1]}
       \left\{\OPT(P_\theta)-\E_\theta|M_n^A|\right\}.
\]
Theorem~\ref{thm:epoch} implies
\[
 \mathfrak R_n=O(\sqrt n\,\log^2 n)=o(n^{2/3}).
\]
Thus $n^{2/3}$ is not the minimax optimal learning rate in this
constant-parameter model. The exact minimax rate, and whether the
logarithmic factors can be removed, remain open. This conclusion
does not assert that the pilot-and-commit policy itself incurs
$\Theta(n^{2/3})$ regret, or that the geometric-update policy
outperforms it at every finite horizon.

\FloatBarrier
\section{Reproducible computations}\label{sec:computations}

All reported values are obtained by deterministic recursions. No Monte Carlo estimates are used. Exact small-instance checks use rational arithmetic; larger computations use double precision.

An independent Bellman recursion branches on parent degree and acceptance or rejection, without assuming separability. For $n\in\{4,6,8,10\}$ and $\theta\in\{0,1/2,1\}$, it agrees with Proposition~\ref{prop:bellman} in all 1515 checked nonterminal states after symmetry reduction. Three additional mismatched-parameter cases at $n=10$ agree with direct policy evaluation. A labelled-history recursion verifies Theorem~\ref{thm:robust} on six non-affine laws at $n=7$ and both four-vertex lower examples. For $\varepsilon=1/12$ in Proposition~\ref{prop:sharp}, the values are $17/12$ and $5/4$, with error $1/12$ and regret $1/6$.

For the estimator, an independent recursion over degree multisets reproduces the scalar leaf-count distribution exactly in twelve cases with $k\in\{4,6,8,10\}$ and the same three parameters. We also check the derivative inequalities at 45 rational parameter-size pairs. Figure~\ref{fig:numerics} uses a $21\times21$ true/forecast parameter grid at $n=1000$, together with 25 deterministically computed leaf-count distributions. These finite computations test the implementation and illustrate the estimates; the theorems are established by the preceding proofs.

For the geometric-update policy, 5940 exact price comparisons satisfy Lemma~\ref{lem:sensitivity}. An independent degree-multiset recursion checks both degree-moment bounds in 75 parameter-size cases, and exact leaf laws check the squared-error estimator bound in 20 cases. Full policy-state evaluation verifies the performance-difference identity for $n\in\{4,6,8,10,12\}$ and $\theta\in\{0,1/2,1\}$, using 2529 cached states across the 15 instances. The streaming implementation is also checked against rational decisions on all 5040 labelled histories at $n=8$ (30240 arrival decisions), and on six longer deterministic histories exercising updates at $8$ and $16$. These are finite correctness checks; they do not establish the asymptotic rate or finite-horizon dominance over other policies.

\begin{table}[htbp]
\caption{Expected numbers of matched edges for $n=1000$, rounded to three decimals. The last column uses the same forecast $\widehat\theta=1$ for every true parameter. These are finite-horizon values under the seed-edge convention.}\label{tab:values}
\centering
\begin{tabular}{rrrr}
\toprule
True $\theta$ & Oracle optimum & Greedy & Forecast $\widehat\theta=1$\\
\midrule
0 & 333.333 & 333.333 & 326.664\\
0.25 & 319.190 & 318.208 & 316.527\\
0.50 & 303.415 & 300.067 & 302.627\\
0.75 & 283.675 & 277.911 & 283.557\\
1 & 257.523 & 250.250 & 257.523\\
\bottomrule
\end{tabular}
\end{table}

\begin{figure}[htbp]
\centering
\includegraphics[width=\linewidth]{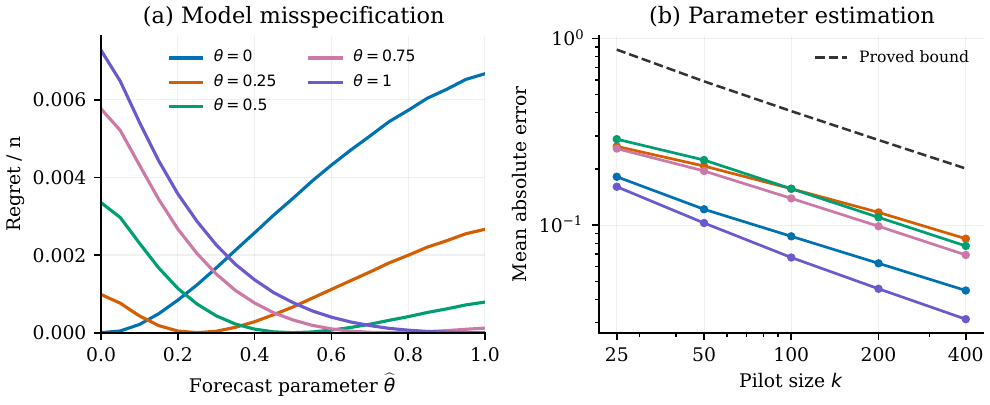}
\caption{Left: plug-in regret per vertex at $n=1000$ as the forecast parameter varies. Right: mean absolute error of the clipped inverse leaf-count estimator, evaluated with bisection tolerance $10^{-8}$; the dashed curve is Corollary~\ref{cor:estimation}. Each colour denotes a fixed true parameter. The estimator curves evaluate the full scalar leaf-count law, rather than sampled trees.}\label{fig:numerics}
\end{figure}

Table~\ref{tab:values} shows why correct-model thresholds and misspecification must be assessed separately. The general parameter bound can substantially exceed the observed loss inside this particular family. No asymptotic rate in $|\theta-\widehat\theta|$ is inferred from these numerical values.

\section{Discussion}

The main structural estimate is a unit-span conditional continuation score for affine attachment forecasts. It gives regret at most twice cumulative conditional model error under arbitrary exogenous misspecification. Within the uniform--preferential family, sensitivity of the individual Bellman prices also supports repeated estimation from a single growing tree.

Geometric updating yields $O(\sqrt n\,\log^2 n)$ regret with $O(n^2\log n)$ arithmetic work and $O(n)$ storage. The exact optimal rate remains unresolved: no matching lower bound of order $\sqrt n$ is proved here, and the logarithmic factors may be artefacts of the uniform sensitivity and degree-moment bounds. The asymptotic comparison also does not imply finite-horizon dominance of one policy over the other.

A finer analysis could exploit how often the process visits states where the decision margin $\Delta_t(d;\theta)$ is close to zero. This may improve both parameter calibration and learning guarantees. The four-vertex lower constructions address general history-dependent misspecification and do not supply a lower bound in the constant-parameter family. Finally, comparison with offline maximum matchings requires additional competitive analysis beyond oracle regret.

\section*{Code and data availability}
The Python implementations, exact verification scripts, deterministic result files, and figure-generation code are available at
\url{https://github.com/mgalazka84/robust-online-matching}.
The version accompanying this manuscript is commit
\href{https://github.com/mgalazka84/robust-online-matching/commit/fbabda8f99986032a8870af844d4ca51f57722ce}{\texttt{fbabda8f9998}}.
The repository includes a documented command for rerunning the checks and regenerating the figure. All reported computations are deterministic; no Monte Carlo samples are used.

\section*{Declaration of generative AI use}
Generative AI tools assisted with mathematical exploration, code development, and drafting.

\begingroup
\interlinepenalty=10000
\bibliography{references}
\endgroup
\end{document}